\documentclass[12pt]{iopart}

\usepackage{amssymb}
\usepackage{epsfig}
\usepackage{color}
\usepackage{adjustbox}
\usepackage{float}
\usepackage{verbatim}
\usepackage[bookmarks]{hyperref}
\usepackage{amssymb,amsfonts,latexsym,amsthm,wrapfig,graphicx, hyperref,  
}
\usepackage{color}

\usepackage{amsfonts}
\usepackage{amssymb}
\usepackage{enumerate}
\usepackage{amsthm}
\usepackage{graphicx}
\usepackage[english]{babel}
\usepackage{hyperref}

\newtheorem{theorem}{Theorem}[section]
\newtheorem*{theorem*}{Theorem}

\newtheorem*{corollary*}{Corollary}
\newtheorem*{claim*}{Claim}

\newtheorem{fact}[theorem]{Fact}
\newtheorem{proposition}[theorem]{Proposition}

\theoremstyle{definition}

\begin{document}

\title{Newton's identities and positivity of trace class integral operators}

\author{G.~Homa$^{1,2}$, R.~Balka$^3$, J.~Z.~Bern\'ad$^4$, M.~K\'aroly$^5$ and A.~Csord\'as$^6$}

\address{$^1$ Wigner Research Centre for Physics, Konkoly-Thege M.~\'ut 29-33, H-1121 Budapest, Hungary}
\address{$^2$ Department of Physics of Complex Systems, E\"{o}tv\"{o}s Lor\'and University, ELTE, P\'azm\'any P\'eter s\'et\'any 1/A, H-1117 Budapest, Hungary}
\ead{ggg.maxwell1@gmail.com} 

\address{$^3$ Alfr\'ed R\'enyi Institute of Mathematics, Re\'altanoda utca 13-15., H-1053 Budapest, Hungary}
\ead{balka.richard@renyi.hu}

\address{$^4$ Peter Gr\"unberg Institute (PGI-8), Forschungszentrum J\"ulich, D-52425 J\"ulich, Germany}
\ead{j.bernad@fz-juelich.de}

\address{$^5$ Balasys IT Zrt. 1117. Budapest, Aliz utca 4., Hungary}
\ead{miklos.karoly@protonmail.ch}

\address{$^6$ Department of Physics of Complex Systems, E\"{o}tv\"{o}s Lor\'and University, ELTE, P\'azm\'any P\'eter s\'et\'any 1/A, H-1117 Budapest, Hungary}
\ead{csordas@tristan.elte.hu}

\date{\today}

\begin{abstract}
We provide a countable set of conditions based on elementary symmetric polynomials that are necessary and sufficient for a trace class integral operator to be positive semidefinite, which is an important cornerstone for quantum theory in phase-space representation. We also present a new, efficiently computable algorithm based on Newton's identities. Our test of positivity is much more sensitive than the ones given by the linear entropy and Robertson-Schr\"odinger's uncertainty relations; our first condition is equivalent to the non-negativity of the linear entropy.
\end{abstract}
\vspace{2pc}

\submitto{\JPA}

\maketitle

\section{Introduction}

Quantum systems are described in terms of density operators, or, in mathematical language, positive trace class operators with trace one \cite{Neumann}. In infinite-dimensional Hilbert spaces, this is a rather abstract object, but with the help of the phase-space representation, the density operator becomes a so-called quasi-probability distribution. The first of these was introduced by Wigner \cite{Wigner}. However, the concept of phase-space representation of a self-adjoint operator was already proposed by Weyl a few years earlier \cite{Weyl}, which he called Hermitian forms. The power of this method was first demonstrated by Moyal \cite{Moyal}: later it has found many applications in quantum chemistry, statistical mechanics, and quantum optics \cite{Hillery, Lee, Schleich, Weinbub}.

Dynamics in the phase-space representation result in partial differential equations, e.g.~the classical Liouville equation for the Wigner function, therefore these exact equations are successfully used for descriptions of open quantum systems \cite{book1}, like the quantum Brownian motion \cite{HuPazZhang92,Halliwell2012}. However, these exact equations are usually subject to further assumptions, which may lead to violations of the positivity of the density operator \cite{Gnutzmann}. Testing of these positivity violations is usually hard in the phase-space representation \cite{BLH}. This is an essential problem for the consistency check of different models, nonetheless, from the foundational point of view of quantum mechanics the characterization of positivity with the so-called KLM conditions has already started in the $1960$s \cite{Kastler, Loupias1, Loupias2}.

Further studies on trace-class operators in phase-space representation have been carried out \cite{Narkowich1, Narkowich2, Werner, Luef}, but the positivity of the operator was usually provided by a non-countable set of conditions. Recently, a countable set of conditions with the help of Gabor frames was found \cite{2019}, where one needs to test the positivity of matrices in which entries are calculated with the help of a lattice structure. In this article we also provide a countable set of conditions, which are necessary and sufficient for the positivity of a self-adjoint trace class operator. Furthermore, they require a tractable computational process, which we demonstrate by examples.

The paper is organized as follows. In Sec.~\ref{sec:II} we establish notations and our main results, which are obtained by using properties of trace class operators and elementary symmetric polynomials. We apply the derived set of conditions to different examples in Sec.~\ref{sec:III} and compare them with some frequently used simple tests. In Sec.~\ref{sec:IV} we summarize and draw our conclusions.

\section{Theoretical and mathematical background }
\label{sec:II}

Recall that $L^2(\mathbb{R}^n)$ denotes the Hilbert space of complex-valued square-integrable functions defined on $\mathbb{R}^n$. We consider Hilbert-Schmidt operators $\hat{\rho}$ in the form 
\begin{equation}
 \left ( \hat{\rho} f \right) (x) = \int_{-\infty}^\infty \, \rho(x,y) f(y)\, \mathrm{d}y, \label{eq:densitydef}
\end{equation}
where $\rho \in L^2(\mathbb{R}^2)$ is the kernel and $f \in L^2(\mathbb{R})$, see \cite{book3}. 
The self-adjointness property is necessary for the positivity of $\hat{\rho}$, and it comes with $\rho(x,y)=\rho^*(y,x)$, where $z^*$ denotes the complex conjugate of $z$, thus we only consider self-adjoint operators from now on. Every Hilbert-Schmidt operator is compact, that is, the closure of the image of the open unit ball under the operator is compact \cite{Rudin}. Therefore, if $\hat{\rho}$ is a (self-adjoint) compact operator then it has only countably many eigenvalues $\{\lambda_n\}^\infty_{n=0}$, see \cite[Theorem~4.25]{Rudin}. The eigenvalue equation of $\hat{\rho}$ is a Fredholm-type integral equation
\begin{equation}
\int_{-\infty}^\infty \, \rho(x,y) \phi_n(y)\,\mathrm{d}y=\lambda_n \phi_n(x).  \label{eq:eigeneq}
\end{equation}
A self-adjoint Hilbert-Schmidt operator $\hat{\rho}$ is trace class if 
\begin{equation}
 \|\hat{\rho}\|_1=\sum_{i=0}^{\infty} |\lambda_i| < \infty. \nonumber
\end{equation}
If $\rho$ is continuous, we also have the formula
\begin{equation}
 \mathrm{Tr}\{\hat{\rho} \}= \sum_{i=0}^{\infty} \lambda_i=\int_{-\infty}^\infty \, \rho(x,x) \, \mathrm{d}x. \label{eq:unittrace}
\end{equation}
We will consider (self-adjoint) trace class integral operators throughout the paper. In quantum mechanical descriptions of physical systems, these eigenvalues are probabilities, thus it is required that $1\geq \lambda_n \geq 0$, i.e.~$\hat{\rho}$ is a positive semidefinite operator, and $\mathrm{Tr}\{\hat{\rho}\}=1$. 

Let $\mathcal{P}=\{p(x)e^{-x^2/2}: p\textrm{ is a complex polynomial from } \mathbb{R} \textrm{ to }\mathbb{C} \}$. As the  weighted polynomials $\{x^n e^{-x^2/2}: n\geq 0\}$ form a basis for $L^2(\mathbb{R})$ (see e.g.~\cite{Johnston}), it follows that $\mathcal{P}$ is dense in $L^2(\mathbb{R})$. Assume that $\hat{\sigma}$ is a positive semidefinite operator with unit trace and kernel $\sigma(x,y)$, and the function $g$ satisfies 
\begin{equation} \label{gf} 
gf\in L^2(\mathbb{R}) \textrm{ for all } f\in \mathcal{P}. 
\end{equation}
We claim that the kernel $\rho(x,y)=g(x)^* \sigma(x,y)g(y)$ defines a positive semidefinite operator $\hat{\rho}$, which can be normalized to have unit trace. Indeed, as the operator $\hat{\rho}$ and the inner product $\langle \cdot , \cdot \rangle$ are continuous in $L^2(\mathbb{R})$ and $\mathcal{P}$ is dense in $L^2(\mathbb{R})$, it is enough to check that $\langle f,\hat{\rho} f \rangle\geq 0$ for all $f\in \mathcal{P}$. Fix an arbitrary $f\in \mathcal{P}$. The positivity of $\hat{\sigma}$ and $gf\in L^2(\mathbb{R})$ imply that
\begin{eqnarray}
\langle f,\hat{\rho} f \rangle&=& \int \!\!\!\!\int_{\mathbb{R}^2} f(x)^* g(x)^* \sigma(x,y) g(y) f(y)\,\mathrm{d}x \, \mathrm{d}y \nonumber\\
&=&\int \!\!\!\! \int_{\mathbb{R}^2} \big[g(x)f(x)\big]^* \sigma(x,y)\big[g(y)f(y)\big]\,\mathrm{d}x \, \mathrm{d}y\ge 0,
\end{eqnarray}
 so $\hat{\rho}$ is positive semidefinite. Hence if $g_j$ satisfy (\ref{gf}) for all $j$, then the convex combinations of the form
\begin{equation}\label{convcomb}
\sum_j \alpha_j g_j(x)^* g_j(y) \sigma(x,y)
\end{equation}
are also positive semidefinite operators, where $\alpha_j\ge 0$ and $\sum_j \alpha_j=1$.

In the case of Schwartz kernels, (that is,~$\rho(x,y)$ and all of its mixed partial derivatives are rapidly decreasing, see \cite[p.~133]{SR} for the precise definition), $\hat{\rho}$ is a trace class operator, see \cite[Proposition~1.1]{Brislawn} and the remark afterwards. Schwartz kernels appear naturally when one studies the density operator of a quantum harmonic oscillator. Note that it is easy to transform our kernel $\rho(x,y)$ to the Wigner function $W(x,p)$ and vice versa:
\begin{equation} \label{eq:Wigner}
W(x,p)=\frac{1}{2 \pi \hbar} \int_{\mathbb{R}} e^{-\frac{i}{\hbar} p y} \rho\left(x+\frac{y}{2},x-\frac{y}{2}\right)\, \mathrm{d}y,   \end{equation}
where $i$ is the imaginary unit and $\hbar$ denotes the reduced Planck constant.

\subsection{Main Theorem} 
To determine positivity, our key tool will be the sequence $e_k$, which is defined as the elementary symmetric polynomials of the eigenvalues $\{\lambda_n\}_{n\geq 0}$:
\begin{equation}
e_k=\sum_{0\leq i_1<\dots<i_k} \lambda_{i_1}\cdots \lambda_{i_k}, \quad \mathrm{for} \quad k \geq 1.
\label{eq:symmetric_polynomials_and_e_k}
\end{equation}
If $\hat{\rho}$ is positive semidefinite operator, then $\lambda_n\geq 0$ for all $n\geq 0$, and it is straightforward that $e_k \geq 0$ for each $k\geq 1$. The reverse implication also holds. Our tests depend on the following important claim:
\begin{proposition} \label{p:claim}
\begin{equation*}
e_k \geq 0  \textrm{ for each }  k \geq 1 ~  \Longrightarrow ~  \lambda_n \geq 0 \textrm{ for all } n \geq 0.
\end{equation*}
\end{proposition}

\begin{proof}
Set $e_0=1$, by \cite[Lemma~3.3]{book3}; we obtain
\begin{equation} \label{eq:x}
\prod_{n=0}^{\infty} (1+\lambda_nx)=\sum_{k=0}^{\infty} e_k x^k \quad  \textrm{is finite for all } x\in \mathbb{R}.
\end{equation}
Although (\ref{eq:x}) is known, for the readers' convenience and to make the proof self-contained, we provide an easier, elementary proof for it, which does not use the theory of complex functions. 

Let us fix an arbitrary real $x$. First, we show that $\prod_{n=1}^{\infty} (1+|\lambda_nx|)<\infty$. By the Taylor expansion of $\log(1+x)$ there exists an $0<\varepsilon<1$ such that $|\log(1+|x|)|<2|x|$ whenever $|x|<\varepsilon$. Let $N=N(x,\varepsilon)$ be a sufficiently large positive integer such that $|\lambda_n x|<\varepsilon$ for all $n\geq N$. Clearly, it is enough to prove that $\prod_{n\geq N} (1+|\lambda_nx|)$ is finite. We can estimate its logarithm as
\begin{eqnarray*}
\left| \log\left( \prod_{n\geq N} (1+|\lambda_nx|)\right) \right|&\leq \sum_{n\geq N} |\log (1+|\lambda_n x|)|
\\ &\leq \sum_{n\geq N} 2|\lambda_n x|= 2|x| \sum_{n\geq N}|\lambda_n|<\infty,
\end{eqnarray*}
hence $\prod_{n=1}^{\infty} (1+|\lambda_nx|)$ is finite.

After the expansion $\prod_{n=1}^{m+1}  (1+\lambda_nx)$ contains all terms of $\prod_{n=1}^{m}  (1+\lambda_nx)$,  so the product $\prod_{n=1}^{\infty} (1+\lambda_nx)=\lim_{m\to \infty} \prod_{n=1}^m (1+\lambda_nx)$ is an infinite series by definition. Moreover, $\prod_{n=1}^{\infty} (1+\lambda_nx)$ and $\sum_{k=0}^{\infty} e_k x^k$ are the same series with rearranged terms. Since $\prod_{n=1}^{\infty} (1+|\lambda_nx|)$ is finite, these series are absolutely convergent, so they are equal and finite. This implies (\ref{eq:x}).

Now, we can finish our proof. Assume to the contrary that there is an integer $m\geq 0$ such that $\lambda_{m}<0$. Let $x_0=-1/\lambda_{m}$, then clearly $x_0>0$ and $\prod_{n=0}^{\infty} (1+\lambda_nx_0)=0$. Eq.~(\ref{eq:x}) implies that $\sum_{k=0}^{\infty} e_k x_0^k=0$. Using that $x_0>0$, $e_0=1$, and $e_k \geq 0$ for all $k \geq 1$, we obtain that $\sum_{k=0}^{\infty} e_k x_0^k \geq  1$, which is a contradiction. The proof is complete.
\end{proof}
Thus, when $e_k \geq 0$ for each $k \geq 1$ then $\hat{\rho}$ is a positive semidefinite operator. It is worth noting that $e_k$s have been used for the positivity test of $n \times n$ self-adjoint matrices, where it is enough to check $n-1$ conditions \cite{Gamel}.   

\subsection{Newton's identities and a useful estimation}
Now, we need to express all the $e_k$s with the help of the kernel $\rho(x,y)$. We calculate the moments $M_k$ of $\hat{\rho}$ as
\begin{eqnarray} \label{n-th_mom}
M_k = \sum_{i=0}^{\infty}\lambda_i^k=\mathrm{Tr}\{ \hat{\rho}^k\}
=\int_{-\infty}^\infty \, \rho(x_k,x_1) \prod_{i=1}^{k-1} \rho(x_i,x_{i+1}) \prod_{i=1}^k \mathrm{d}x_i.
\end{eqnarray}

Despite the countably infinite number of eigenvalues, we can still obtain Newton's identities \cite{Macdonald} in the form
\begin{equation}
e_k =\frac{1}{k!}\left|\begin{array}{ccccc}
M_1     & 1       & 0      & \cdots       \\
M_2     & M_1     & 2      & 0      & \cdots \\
\vdots  &         & \ddots & \ddots       \\
M_{k-1} & M_{k-2} & \cdots & M_1    & k-1 \\
M_k     & M_{k-1} & \cdots & M_2    & M_1
\end{array}\right|. \label{eq:e_n_from_det}
\end{equation}
We note that an equivalent option is to use the Fredholm expansion \cite[Theorem~3.10]{book3} (for the original source see \cite{Fredholm}). Furthermore, in case of $M_1=1$ the quantity $2 e_2=1-M_2$ is called linear quantum entropy in the literature, see for example \cite{LQE}. The sequence $e_k$ rapidly converges to zero. Indeed, expanding $\left(\sum_{n=0}^{\infty} |\lambda_n|\right)^k$ yields the following inequality, see \cite[Lemma~3.3 (3.4)]{book3}:
\begin{equation}
|e_k|\leq \frac{\left(\sum_{n=0}^{\infty}|\lambda_n| \right)^k}{k!} \quad \textrm{for all } k\geq 1.
\label{eq:superexp_behav}
\end{equation}
 Finally, it is worth to mention that for $\rho(\mathbf{x}, \mathbf{y}) \in L^2(\mathbb{R}^{2n})$ with $\mathbf{x}=(x_1,x_2,...,x_n)^{T}$ and $\mathbf{y}=(y_1,y_2,...,y_n)^{T}$ ($T$ denotes the transpose of vectors) our proof works verbatim.

\section{Examples}
\label{sec:III}

In this section, we apply our method to various examples. First, we look at the known case of Gaussian quantum states and then at different kernels $\rho(x, y)$ in the form of a polynomial multiplied by a Gaussian function. We also do a comparison with other approaches like the physically motivated Robertson-Schr\"odinger uncertainty relations. Throughout the entire section, we omit physical dimensions.

\subsection{The Gaussian case: a reminder}
One of the simplest examples of a density operator is the Gaussian quantum state
\begin{eqnarray}
\rho_G(x,y)&=&2\sqrt{\frac{C}{\pi}}\exp\Bigl[-\Bigl(A(x-y)^2+iB(x^2-y^2)+
 \Bigr.\Bigr.  \nonumber\\
&&\Bigl.\Bigl. +C(x+y)^2+iD(x-y)+E(x+y)+\frac{E^2}{4C}\Bigl)\Bigr],
\end{eqnarray}
with real parameters $A>0,C>0,B,D,E$. It can be checked easily that $\textrm{Tr}\{ \hat{\rho}\}=1$. Orthonormalized eigenvectors and eigenvalues are given in \cite{BCsH} (after correcting some minor errors):
\begin{eqnarray}
\phi_n(x)&=&\sqrt{\frac{2(AC)^{1/4}}{\sqrt{\pi} 2^n n!}} H_n\left(2(AC)^{1/4}\Bigl(x+\frac{E}{4C}\Bigr)\right) \nonumber \\
&&\times \exp\left[-x^2\Bigl(2\sqrt{AC}+iB\Bigr)-x\Bigl(\sqrt{\frac{A}{C}}E+iD \Bigr)-\frac{\sqrt{AC}}{8C^2}E^2\right]
\label{eq:gaussian_ef}
\end{eqnarray}
and
\begin{equation}
\lambda_n=\epsilon_0 \epsilon^n,
\end{equation}
where we used the notations
\begin{equation}
\epsilon_0=\frac{2\sqrt{C}}{\sqrt{A}+\sqrt{C}}, \quad \epsilon=\frac{\sqrt{A}-\sqrt{C}}{\sqrt{A}+\sqrt{C}}, \quad r=2(AC)^{1/4}, \quad s=\frac{E}{4C}.
\label{eq:gaussian_spectrum}
\end{equation}
Here $H_n$ is the $n$-th Hermite polynomial. From the spectrum of $\hat{\rho}_G$ it is clear that $\hat{\rho}_G$ is a positive operator precisely if
\begin{equation}
 A \geq C >0.
\label{eq:allowed_gaussian_params}
\end{equation}
The moments defined in (\ref{n-th_mom}) are given by
\begin{equation}
M_k=\frac{\epsilon_0^k}{1-\epsilon^k},
\end{equation}
and one can check that all the $e_k$s are strictly positive if and only if $A\geq C$.

\subsection{Linear polynomials multiplied by a Gaussian}

More interesting behavior of $e_k$s can be exhibited if $\rho_G(x,y)$ is multiplied by a self-adjoint polynomial with real coefficients and variables $x,y$. First, we consider here the case of linear polynomials:
\begin{equation*}
\rho(x,y)=(\alpha_1(x+y)+i\beta_1(x-y)+\gamma_0) \rho_G(x,y), \quad \textrm{where } (\alpha_1,\beta_1)\neq (0,0).
\end{equation*}

By calculating $e_k$s and applying Proposition~\ref{p:claim} in special cases we predicted that there exists no positive operator of the above form, see Fig.~\ref{fig:lin_gauss_path} for illustration: We define the functions $\{ \Theta_k\}_{k\geq 1}$ such that $\Theta_k(t)=1$ if $e_i(t)\geq 0$ for all $1\leq i\leq k$ and $\Theta_k(t)=0$ otherwise. Therefore, $\Theta_k$ is the indicator function of the set of parameters $t$ for which $e_i(t)\geq 0$ for all $1\leq i\leq k$, which seem to form rapidly decreasing intervals as $k\to \infty$. This suggests that $\{t: e_i(t)\geq 0 \textrm{ for all } i\geq 1\}=\emptyset$. The following theorem shows that this is indeed the case.

\begin{figure}
\includegraphics[angle=270,width=1.\linewidth]{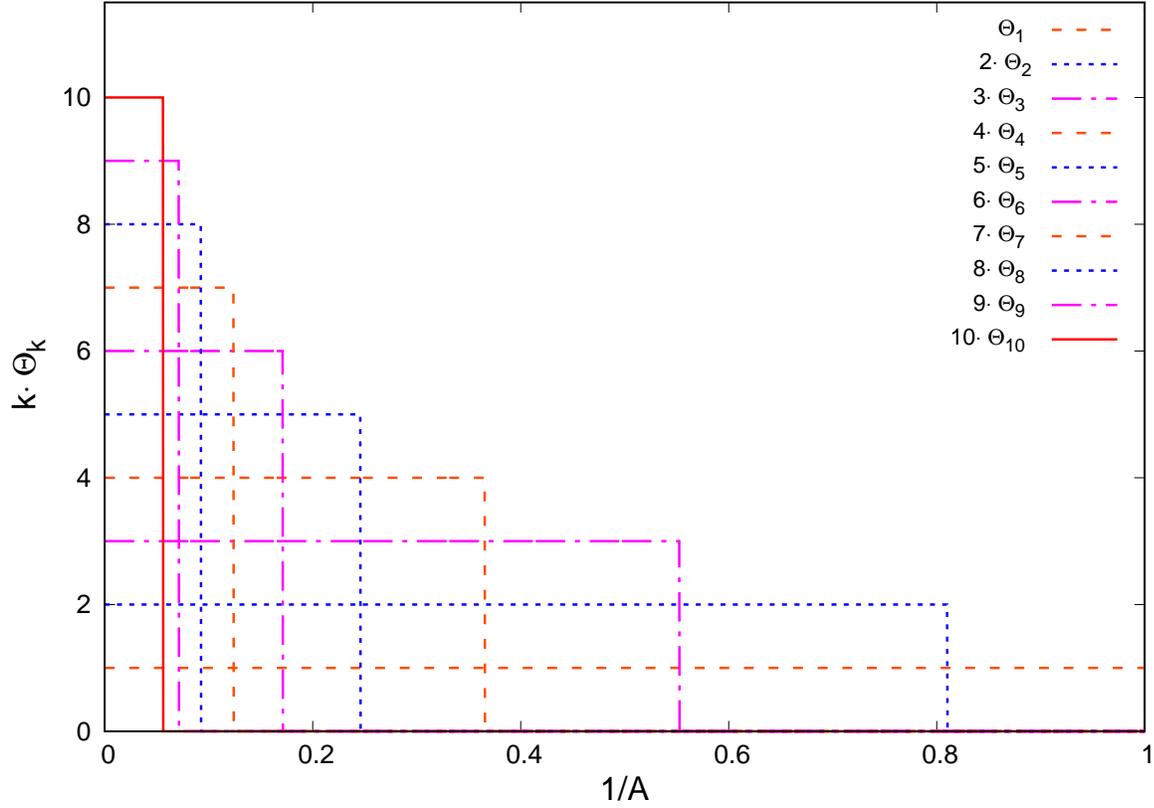}
\caption{Several $\Theta_k$s calculated for $\rho(x,y)$ of Eq.~(\ref{eq:rho_eq_gaussian_times_poly}) as functions of the Gaussian parameter $1/A$.
The parameters of the Gaussian are: $C=1$, $B=D=0$ and $E=1$. The polynomial parameters are set to $\alpha_1=1$, $\beta_1=0$, $\gamma_0=2$, and $\alpha_2=\beta_2=\gamma_2=0$.
\label{fig:lin_gauss_path}}
\end{figure}

\begin{proposition} Assume that
$$\rho(x,y)=\left(\alpha_1(x+y)+i\beta_1(x-y)+\gamma_0\right)\rho_G(x,y)$$
such that $(\alpha_1,\beta_1)\neq (0,0)$. Then $\hat{\rho}$ is not positive semidefinite.
\end{proposition}
\begin{proof}
Define the complex, square integrable function
\begin{equation*}
    \Psi(x):=\exp\left({-x^2+bx+icx}\right)\times \exp\left({-iBx^2+Ex-iDx}\right),  
    \end{equation*}
where $b,c \in \mathbb{R}$.
It is sufficient to show that the integral
\begin{equation*}
\Pi=\Pi(b,c):=\int_{\mathbb{R}^2}{\Psi^*(x)\Psi(y)\rho(x,y) \, \mathrm{d} x \, \mathrm{d}y}
\end{equation*}
can attain negative values for some parameters $b,c$. We can calculate that
\begin{eqnarray*}
\Pi=2 \sqrt{\pi C} e^{-\frac{E^2}{4C}} \left( (2A+1)(2C+1) \right) ^{-3/2} e^{\frac {2Ab^2-2Cc^{2}+b^{2}-c^{2}}{2(4AC+2A+2C+1)}}
  \\
\times \left((2A+1)\alpha _{1}b+(2C+1)\beta _{1}c+ (2A+1)(2C+1)\gamma _{0}\right).
\end{eqnarray*}

Note that only the last factor of $\Pi$ might be non-positive. Since $(\alpha_1, \beta_1)\neq (0,0)$, this factor is a non-constant linear polynomial in the variables $b$ and $c$, which can clearly attain negative values. This completes the proof. 
\end{proof}

\subsection{Quadratic polynomials multiplied by a Gaussian}

Now we consider second degree, self-adjoint polynomials multiplied by a Gaussian, in which case we will obtain more sophisticated behaviour from the point of view of positivity. Define
\begin{eqnarray} \label{eq:rho_eq_gaussian_times_poly}
\rho(x,y)&=&\frac 1N \rho_G(x,y) \Bigl( \alpha_2(x-y)^2+i\beta_2(x^2-y^2) \Bigl. \nonumber\\
&&\Bigl. +\gamma_2(x+y)^2+\alpha_1(x+y)+i\beta_1(x-y)+\gamma_0\Bigr),
\end{eqnarray}
where
\begin{equation*}
N=\gamma_0+\displaystyle\frac{\gamma_2-\alpha_1E}{2C}+\displaystyle\frac{\gamma_2 E^2}{4C^2}
\end{equation*}
is a normalization factor ensuring $\textrm{Tr}\{\hat{\rho}\}=1$. 

The calculation of the quantities $e_k$ can be done directly through Eqs.~(\ref{n-th_mom}) and (\ref{eq:e_n_from_det}). In case of Eq.~(\ref{eq:rho_eq_gaussian_times_poly}), there is an alternative way to obtain the moments $M_k$. One needs to calculate the matrix elements $\rho_{m,n}=\langle \phi_m, \hat{\rho}\phi_n \rangle$ of $\hat{\rho}$ (see Eq.~(\ref{eq:rho_eq_gaussian_times_poly})) between the states $m,n$ explicitly given by (\ref{eq:gaussian_ef}). By repeated use of the well-known recursion for Hermite polynomials
\begin{equation}
    H_{n+1}(x)=2x H_n(x)-2n H_{n-1}(x),
\end{equation}
it turns out that $\rho_{m,n}$ is a band matrix with two subdiagonals below and above the diagonal, that is,  $\rho_{m,n}=0$ if $|m-n|>2$. Explicit form for the matrix elements can be obtained as
\begin{eqnarray}
\rho_{m,n}&=&\sqrt{n(n-1)}\epsilon^{n-2}(a_2 -i b_2)\delta_{n,m+2}+\sqrt{n}\epsilon^{n-1}(a_1 -i b_1)\delta_{n,m+1} \nonumber\\
&&+\epsilon^n (a_0+n b_0)\delta_{n,m}+\sqrt{n+1}\epsilon^{n}(a_1 +i b_1)\delta_{n,m-1} \nonumber \\
&&+\sqrt{(n+1)(n+2)}\epsilon^{n}(a_2 +i b_2)\delta_{n,m-2},
\end{eqnarray}
where $\delta_{n,m}$ is the Kronecker delta and the quantities $a_2,b_2,a_1,b_1,a_0$ and $b_0$ are real-valued constants. They depend linearly on the polynomial parameters $\alpha_2,\beta_2,\gamma_2,\alpha_1,\beta_1,\gamma_0$, but non-linearly on the Gaussian parameters used in Eq. (\ref{eq:gaussian_spectrum}). Namely:
\begin{eqnarray}
a_0&=&\frac{1}{r^2}\left((1-\epsilon)\alpha_2+(1+4 r^2s^2 +\epsilon) \gamma_2-2 r^2 s \alpha_1 +r^2 \gamma_0 \right), \nonumber \\
b_0&=&\frac{\epsilon_0}{r^2\epsilon}\left(-(1-\epsilon)^2\alpha_2+(1+\epsilon)^2 \gamma_2 \right), \nonumber \\
a_1&=& -\frac{\epsilon_0(1+\epsilon)}{\sqrt{2}r} \left(4 s \gamma_2-\alpha_1\right), \nonumber \\
b_1&=& \frac{\epsilon_0(1-\epsilon)}{\sqrt{2}r} \left(2s\beta_2-\beta_1\right),\nonumber \\
a_2&=& \frac{\epsilon_0}{2 r^2}\left((1-\epsilon)^2\alpha_2+(1+\epsilon)^2\gamma_2 \right), \nonumber \\
b_2&=& -\frac{\epsilon_0(1-\epsilon^2)}{2 r^2} \beta_2.
\end{eqnarray}
After finding the matrix elements, the summation over the diagonal elements of $k$-th power of this latter matrix yields $M_k$. Fortunately, this expression involves some combinations of the form $\sum_{n=0}^\infty n^p \epsilon^{nq}$ with integers $p,q$ and $\epsilon$ from Eq.~(\ref{eq:gaussian_spectrum}), which can be given explicitly. In fact, this serves as a validation of the numerical evaluation of Eq.~(\ref{n-th_mom}).

\subsection{Testing our method for families where positivity is understood}

From now on, our kernels will be quadratic polynomials multiplied by a Gaussian. First, we apply our method in a special situation, where $\hat{\rho}$ is a positive semidefinite operator in the entire parameter region. In Fig.~\ref{fig:ek_are_good} we use $\rho(x,y)\propto (4xy+1)\rho_G(x,y)$. As $g(x)=x$ satisfies property (\ref{gf}), we obtain that $\rho(x,y)$ is proportional to a convex combination of the form (\ref{convcomb}), hence positive. In this case we fix the Gaussian parameter $A=1$, and plot some $e_k$s in the region of positivity allowed by Eq.~(\ref{eq:allowed_gaussian_params}). As expected, we obtained positive values for all calculated $e_k$s in the region $0<\sqrt{C}\leq  1$, see Fig.~\ref{fig:ek_are_good}.
\begin{figure}
\includegraphics[angle=270,width=1.\linewidth]{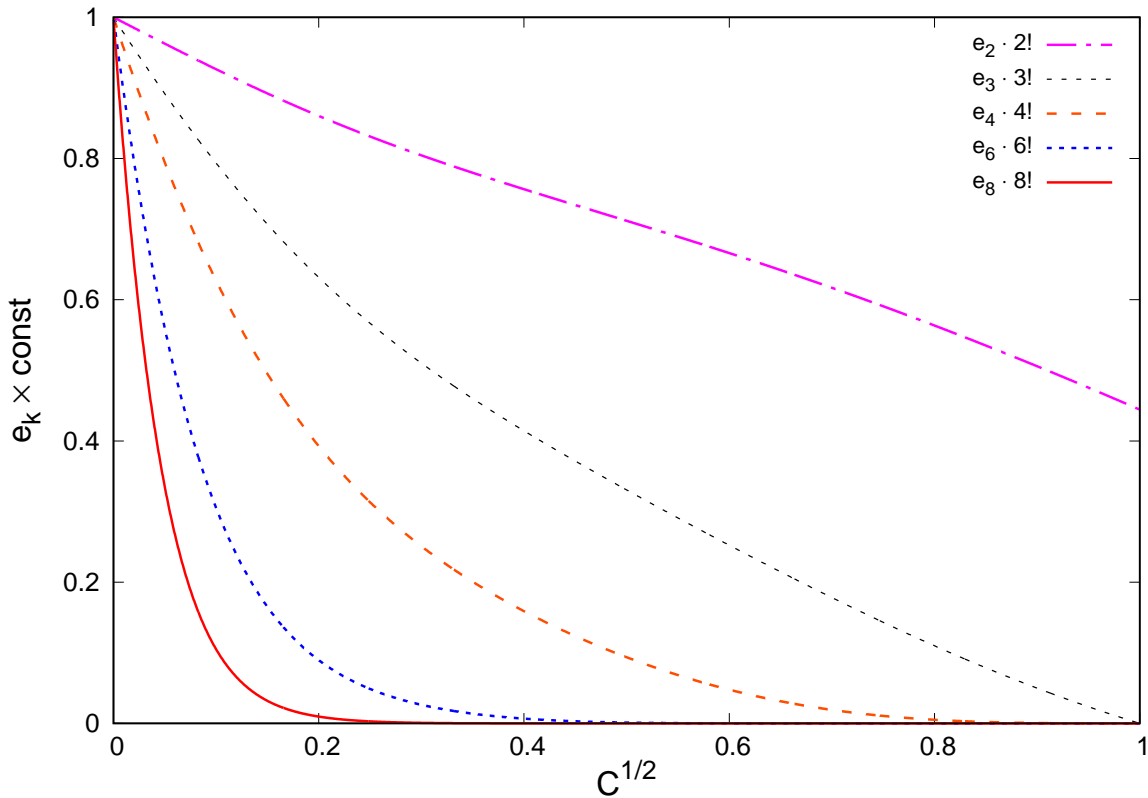}
\caption{Several $e_k$s as functions of $\sqrt{C}$. For better visualization $e_k$s are scaled by $k!$. The parameters of the Gaussian are $A=1$, $B=D=E=0$, the polynomial parameters are chosen to be: $\alpha_2=-1$, $\gamma_2=\gamma_0=1$, $\alpha_1=\beta_1=\beta_2=0$.  $\rho(x,y)\propto (4xy+1)\rho_G(x,y)$.
\label{fig:ek_are_good}}
\end{figure}

Now, we examine the family $\rho(x,y)\propto (4xy+\gamma_0)\exp{\left[-\left(4(x-y)^2+(x+y)^2\right)\right]}$. We show that $\hat{\rho}$ is positive semidefinite if and only if $\gamma_0\geq 0$. Indeed, if $\gamma_0\geq 0$ then $\hat{\rho}$ is positive semidefinite by (\ref{convcomb}). Applying the following fact for $z=0$ shows that $\hat{\rho}$ is not positive semidefinite if $\gamma_0<0$. 
\begin{fact} Let $\rho$ be a kernel and $z\in \mathbb{R}$ such that $\rho$ is continuous at $(z,z)$. If $\hat{\rho}$ is positive semidefinite, then $\rho(z,z)\geq 0$.
\end{fact}
\begin{proof}
Assume to the contrary that $\rho(z,z)<0$ and $\hat{\rho}$ is positive semidefinite. By the continuity of $\rho$ at $(z,z)$ we can choose $\varepsilon>0$ such that $\rho(x,y)<0$ for all $x,y\in [z-\varepsilon,z+\varepsilon]$. Define the square integrable function $\Psi$ such that $\Psi(x)=1$ if $z-\varepsilon \leq x \leq z+\varepsilon$ and $\Psi(x)=0$ otherwise. Then clearly 
\begin{equation*}
\langle \Psi, \hat{\rho} \Psi \rangle=\int_{z-\varepsilon}^{z+\varepsilon} \int_{z-\varepsilon}^{z+\varepsilon} \Psi^{*}(x)\Psi(y) \rho(x,y) \,\mathrm{d}x \, \mathrm{d}y<0. 
\end{equation*}
This contradicts that $\hat{\rho}$ is positive semidefinite, which concludes the proof. 
\end{proof}

\begin{figure}
\includegraphics[angle=270,width=1.\linewidth]{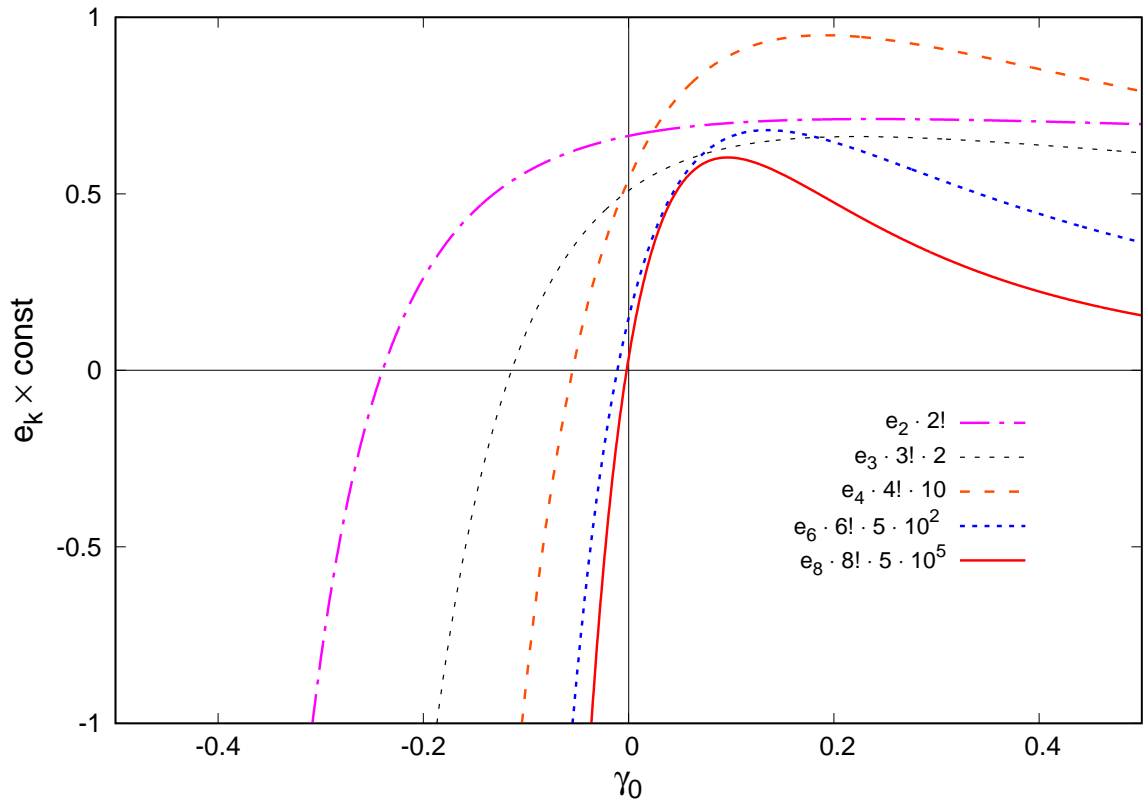}
\caption{Several $e_k$s calculated for $\rho(x,y)$ of Eq.~(\ref{eq:rho_eq_gaussian_times_poly}) as functions of the polynomial parameter $\gamma_0$.
The parameters of the Gaussian are: $A=4$, $C=1$, and $B=D=E=0$. The polynomial parameters are: $\alpha_2=-1$ and $\alpha_1=\beta_1=\beta_2=0$. For $\gamma_0<0$, $\hat{\rho}$ is not a positive operator. We scaled the $e_k$s appropriately for better visualization.\label{fig:gamma_0_path}}
\end{figure}

Let $\rho_p$ be our kernels, where the parameter $p$ runs over a subset of the Euclidian space $\mathbb{R}^d$, and $e_i(p)$ incorporates the parameter dependence of the $e_i$ quantities.    
Our numerical experience is that the sequence of sets 
\begin{equation*} H_k=\{p: e_i(p)\geq 0 \textrm{ for all } 1\leq i\leq k\}
\end{equation*} 
rapidly converges to the parameter space of positivity as $k\to \infty$.  We illustrate this behaviour with Fig.~\ref{fig:gamma_0_path}, where $p=\gamma_0$ and we know that the final set of positivity is $H_{\infty}=[0,\infty)$.

\subsection{Prediction based on our method and comparison with other approaches}

In Figs.~\ref{fig:moments_as_a_function_gamma2} and~\ref{fig:eks_as_a_function_gamma2} the parameters are chosen in such a way that at $\gamma_2=1$ the density operator is known to be positive semidefinite by (\ref{convcomb}). In the following, we also demonstrate that the moments alone (by checking if $M_k>1$, which would imply the existence of a negative eigenvalue) do not reveal too much information about the positivity of $\hat{\rho}$. Indeed, if $M_2\leq 1$, then $M_k\leq M_2\leq 1$ for all $k\geq 2$, so $M_2$ contains all the information. As $1-M_2=2e_2$, this method is equivalent to testing $e_2<0$. In Fig.~\ref{fig:moments_as_a_function_gamma2} we have plotted some moments as functions of the polynomial parameter $\gamma_2$. Note that $M_2>1$ for $\gamma_2 \lesssim -0.65$, which implies the non-positivity of $\hat{\rho}$ in that region.

In Fig.~\ref{fig:eks_as_a_function_gamma2} we have plotted some $e_k$s for the same parameters and note that $e_2<0$ provides the same region $\gamma_2 \lesssim -0.65$ that is given by $M_2>1$ above. Several $e_k$s are negative on the interval $\gamma_2>0$, where the moments do not indicate non-positivity, according to Fig.~\ref{fig:moments_as_a_function_gamma2}. However, negative values for the $e_k$s give us parameters $\gamma_2$, where the corresponding $\hat{\rho}$ is definitely a non-positive operator. The common interval, where all the calculated values $e_k$ are positive is $0 \lesssim \gamma_2 \lesssim 4$, which includes the point $\gamma_2=1$, where  $\hat{\rho}$ is a positive semidefinite operator. The calculation of $e_k$ for big $k$ is not an easy task, because they tend to zero very fast, see Eq.~(\ref{eq:superexp_behav}). To compensate this rapid decay, we have multiplied the quantities $e_k$ with appropriate numbers. Our general observation after several simulations is that if $\rho$ is a second degree polynomial multiplied by a Gaussian as above, then the parameter set of positivity after $k$ tests, i.e.~$\{\gamma: e_i(\gamma)\geq 0 \textrm{ for all } 1\leq i\leq k\}$ form a decreasing and nested sequence of sets. It also seems that these sets (which might be not connected sets in general) rapidly converge to the final set of positivity, namely $\{\gamma: e_i(\gamma)\geq 0 \textrm{ for all } i\geq 1\}$.

\begin{figure}
\includegraphics[angle=270,width=1.\linewidth]{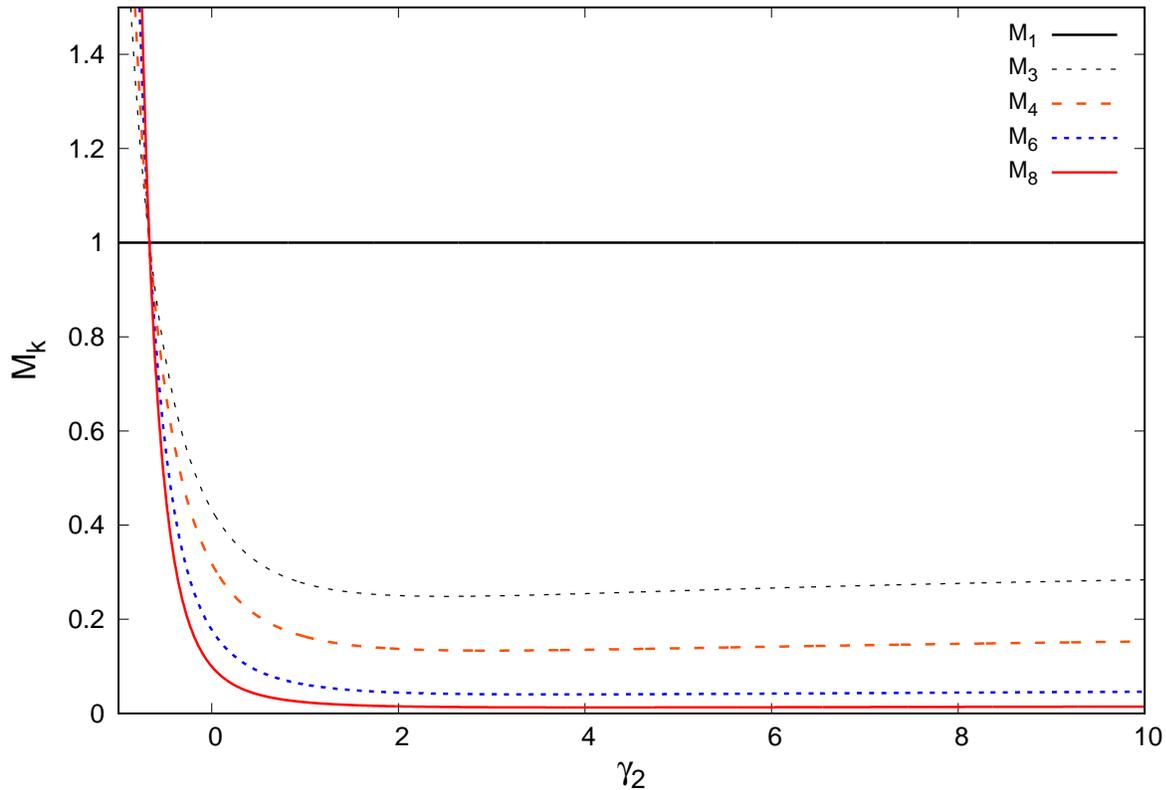}
\caption{Some moments $M_k$ of $\rho(x,y)$ (of the form Eq.~(\ref{eq:rho_eq_gaussian_times_poly})) as a function of the polynomial parameter $\gamma_2$.
The parameters of the Gaussian are $A=3/2$, $C=1$, $B=D=E=0$. The polynomial parameters are $\alpha_2=-1$, $\gamma_0=1$, $\alpha_1=\beta_1=\beta_2=0$. At $\gamma_2=1$, $\hat{\rho}$ is a positive operator. \label{fig:moments_as_a_function_gamma2}}
\end{figure}

\begin{figure}
\includegraphics[angle=270,width=1.\linewidth]{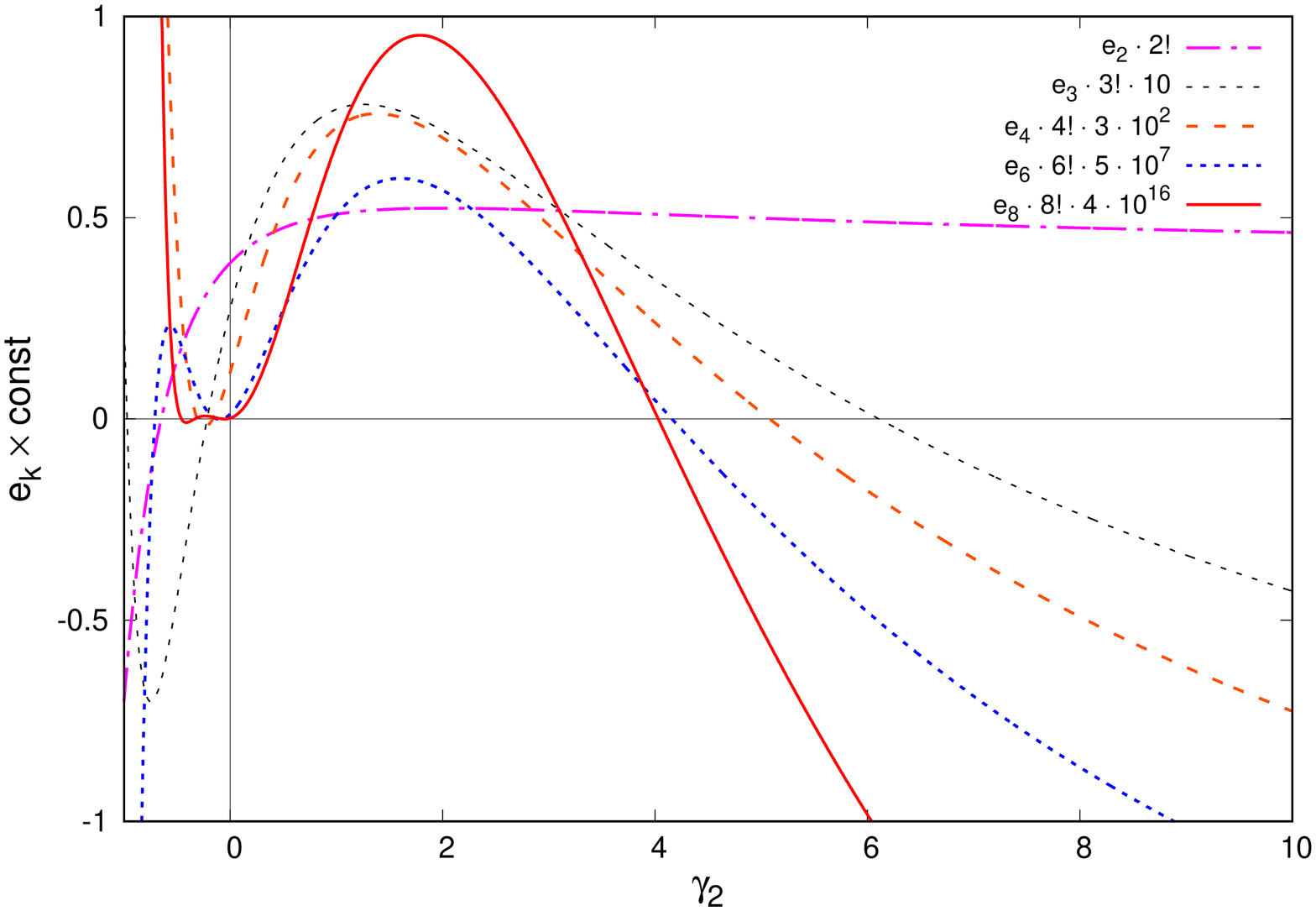}
\caption{Several functions $e_k$ calculated for $\rho(x,y)$ of Eq.~(\ref{eq:rho_eq_gaussian_times_poly}) as functions of the polynomial parameter $\gamma_2$.
The parameters of the Gaussian are $A=3/2$, $C=1$, $B=D=E=0$. The polynomial parameters are $\alpha_2=-1$, $\gamma_0=1$, $\alpha_1=\beta_1=\beta_2=0$. At $\gamma_2=1$, $\hat{\rho}$ is a positive operator. We scaled the $e_k$s appropriately for better visualization.\label{fig:eks_as_a_function_gamma2}}
\end{figure}

\begin{figure}
\includegraphics[angle=270,width=1.\linewidth]{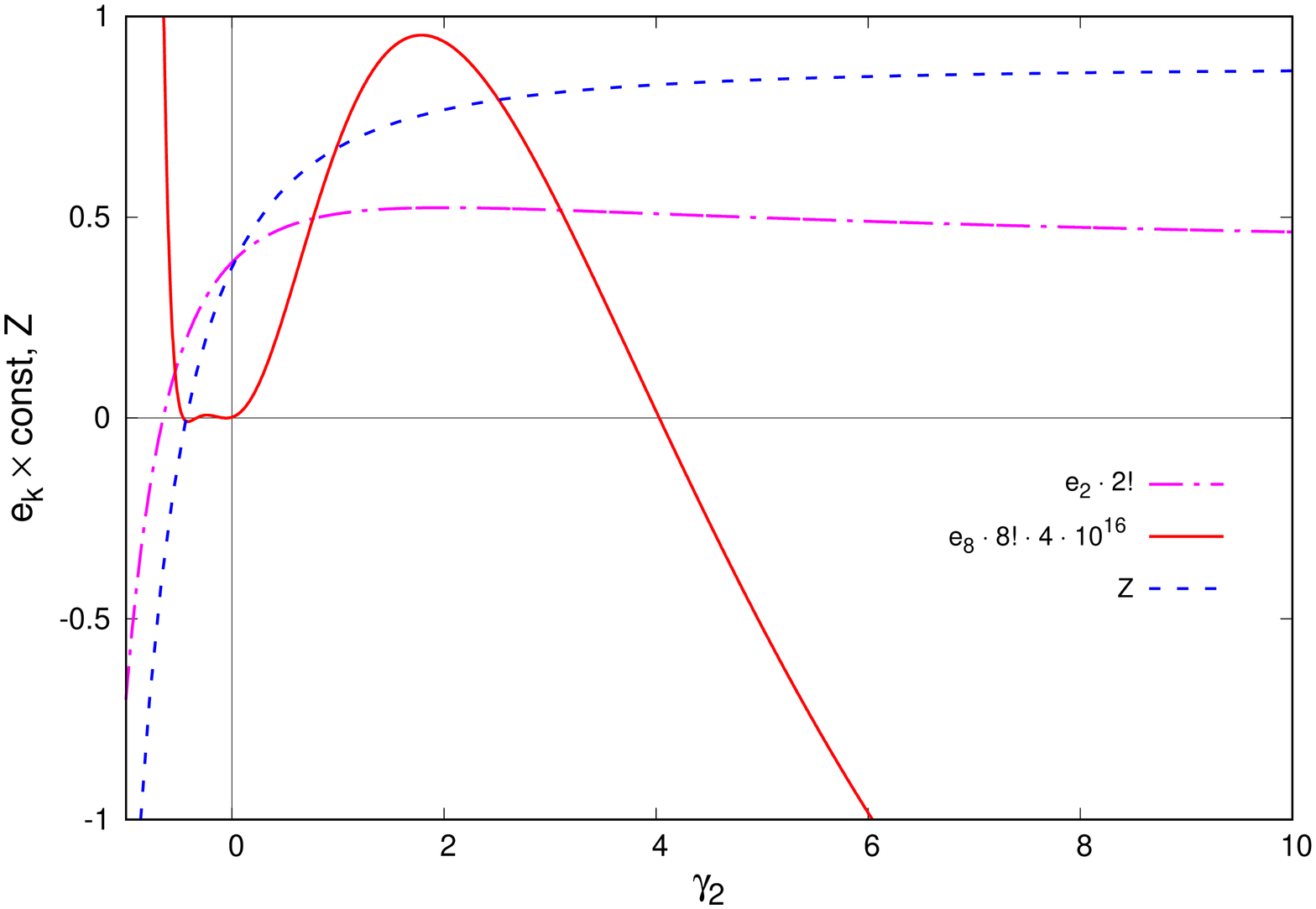}
\caption{Functions $e_2$, $e_8$ and $Z$ calculated for $\rho(x,y)$ of Eq.~(\ref{eq:rho_eq_gaussian_times_poly}) as functions of the polynomial parameter $\gamma_2$.
The  parameters of the Gaussian are $A=3/2$, $C=1$, $B=D=E=0$. The polynomial parameters are $\alpha_2=-1$, $\gamma_0=1$, $\alpha_1=\beta_1=\beta_2=0$. At $\gamma_2=1$, $\hat{\rho}$ is a positive operator. We scaled the $e_k$s appropriately for better visualization.\label{fig:ek2_and_Z_as_afunction_gamma2}}
\end{figure}


Finally, we compare our approach to the Robertson-Schr\"odinger uncertainty relations \cite{Rob, Trif}, which are frequently used to test the positivity of $\hat{\rho}$, see \cite{Flemming}. However, it is known since the $1980$s that fulfilling the uncertainty relations is necessary, but not sufficient, to ensure the positivity of $\hat{\rho}$ \cite{Narkowich1, Luef, Manko}. The uncertainty relations for essentially self-adjoint operators $\hat{A}$ and $\hat{B}$ read as
\begin{equation}
    \sigma_{RS} \ge \frac14 \left|\langle \hat{A}\hat{B}-\hat{B}\hat{A}\rangle \right|^2,
\end{equation}
where $\langle  \hat{O} \rangle=\textrm{Tr}\, \left\{ \hat{\rho}\hat{O}\right\}$.
Here
\begin{equation}
   \sigma_{RS}=\Delta \hat{A}^2 \Delta \hat{B}^2 -\left( \langle \hat{A}\hat{B}+\hat{B}\hat{A} \rangle /2 -\langle  \hat{A} \rangle \langle  \hat{B} \rangle \right )^2,
\end{equation}
where $\Delta \hat{O}^2=\langle  \hat{O}^2 \rangle-\langle  \hat{O} \rangle^2$. For the special choices of $\hat{A}=\hat{x}$ and $\hat{B}=\hat{p}$ we define
\begin{equation}
    Z \equiv \frac{\sigma_{RS}}{ \hbar^2}-\frac14 \ge 0.
\end{equation}
This method tests for $Z\geq 0$, and $Z$ is plotted in Fig.~\ref{fig:ek2_and_Z_as_afunction_gamma2} for the same parameters as in Figs.~\ref{fig:moments_as_a_function_gamma2},\ref{fig:eks_as_a_function_gamma2} together with $e_2$ and $e_8$. This clearly demonstrates that the indicator $Z$ for this choice of $\hat{A}$ and $\hat{B}$ is not much better than $e_2$, and much worse than $e_8$. However, in the special case of $\rho_G(x,y)$, the tests based on the Robertson-Schr\"odinger uncertainty relation and $e_2$ are equivalent, which is well known.

\section{Discussion and conclusions}
\label{sec:IV}

In summary, we have established a computationally tractable method to test the positivity of trace class integral operators via countably many conditions given by Proposition~\ref{p:claim}. A big advantage of our approach is that it extends the method of linear entropy, yet requires only elementary mathematics. In the case of physical applications, phase-space representation is meant to be described by the Wigner function, however, our method requires an extra step, namely the inverse of Eq.~(\ref{eq:Wigner}), which is usually straightforward. 

We have also demonstrated in Section~\ref{sec:III} via several cases that our method is efficient, consistent with the well understood cases, and we can converge rapidly to the interval of parameters where positivity occurs. Furthermore, we showed in Section~\ref{sec:III} that our approach is much more sensitive than the ones given by the methods of  Robertson-Schr\"odinger's uncertainty relation in a special case. 

From a longer-term perspective, our approach can serve as a control for every non-unitary dynamic in the phase-space representation to monitor non-physical evolution. This may apply to unitary dynamic as well, when numerical approximations are applied.

\ack
The authors are indebted to M.~A.~Csirik, M.~Kornyik, Z.~Kaufmann, and \'E.~Papp for helpful discussions. G.~Homa thanks for the funding from the National Research, Development and Innovation Office of Hungary (grants KKP133827 and TKP2021-NVA-04) and acknowledges the support from the Ministry of Innovation and Technology for the Quantum Information National Laboratory. R.~Balka was supported by the MTA Premium Postdoctoral Research Program, the National Research, Development and Innovation Office -- NKFIH, grants no.~124749, and 143285. This paper was supported by the J\'anos Bolyai Research Scholarship of the Hungarian Academy of Sciences. We acknowledge support from NKFI-134437, DFG under Germany's Excellence Strategy-Cluster of Excellence Matter and Light for Quantum Computing (ML4Q) EXC 2004/1-390534769, and AIDAS - AI, Data Analytics and Scalable Simulation - which is a Joint Virtual Laboratory gathering the Forschungszentrum J\"ulich (FZJ) and the French Alternative Energies and Atomic Energy Commission (CEA).

\section*{References}
\label{bib}

\end{document}